\documentclass[11pt,a4paper]{article}
\usepackage{times}
\usepackage{helvet}
\usepackage{courier}

\usepackage{amssymb,latexsym}
\usepackage{graphicx,amsmath,amsfonts,amsthm}
\usepackage{pgf,pgfarrows,pgfnodes,pgfautomata,pgfheaps,pgfshade}
\usepackage[hidelinks]{hyperref}
\usepackage{tikz}
\usepackage{paralist}
\usepackage{color}
\usepackage{comment}
\usepackage[round]{natbib}
\usepackage{breakurl}
\usepackage[]{units}
\usepackage[]{xspace}
\usepackage{url}
\usepackage{booktabs}

\newtheorem{theorem}{Theorem}
\newtheorem{proposition}{Proposition}
\newtheorem{corollary}[proposition]{Corollary}

\newtheorem{definition}{Definition}
\newtheorem{example}{Example}

\newcommand{\naturals}{{{\mathbb{N}}}}

\newcommand{\reals}{{{\mathbb{R}}}}

\newcommand{\np}{{\mathrm{NP}}}

\newcommand{\calR}{\mathcal{R}}

\newcommand{\secref}[1]{Section~\ref{#1}}
\newcommand{\figref}[1]{Figure~\ref{#1}}

\newcommand{\thmref}[1]{Theorem~\ref{#1}}

\newcommand{\propref}[1]{Proposition~\ref{#1}}

\newcommand{\ie}{i.e.,\xspace}
\newcommand{\eg}{e.g.,\xspace}

\newcommand{\phrag}{Phragm\'{e}n\xspace}

\newcommand{\maxP}{max-\phrag}
\newcommand{\varP}{var-\phrag}

\RequirePackage{ifthen,calc}
\newsavebox{\ffbox}\newlength{\ffboxlen}
\newcommand{\todo}[1]{%
  {\sbox{\ffbox}{\textbf{TODO:}\ \textit{{#1}}\ \textbf{:ODOT}}
    \settowidth{\ffboxlen}{\usebox{\ffbox}}
		\addtolength{\ffboxlen}{-5mm}
    \ifthenelse{\ffboxlen>\linewidth}{%
      \noindent\marginpar{$>>>>$}\textbf{TODO:}\ \textit{{#1}}\ \textbf{:ODOT}\marginpar{$<<<<$}}{%
      \noindent\marginpar{$>><<$}\textbf{TODO:}\ \textit{{#1}}\ \textbf{:ODOT}}}}

\allowdisplaybreaks
\title{Multiwinner Approval Rules as Apportionment Methods
}

\date{}

\author{
Markus Brill\\
Technische Universit\"at Berlin\\
brill@tu-berlin.de
\and Jean-Fran\c{c}ois Laslier\\
Paris School of Economics\\
jean-francois.laslier@ens.fr
\and Piotr Skowron\\
  University of Warsaw\\
p.skowron@mimuw.edu.pl 
}

\begin{document}
	
%
\newcommand{\dhondt}{D'Hondt method\xspace} 
\newcommand{\sainte}{Sainte-Lagu\"e method\xspace} 
\newcommand{\hamilton}{largest remainder method\xspace} 
\newcommand{\Hamilton}{Largest remainder method\xspace} 
\newcommand{\HAMILTON}{Largest Remainder Method\xspace} 
%

\maketitle


\begin{abstract}
	We establish a link between multiwinner elections and apportionment problems by showing how approval-based multiwinner election rules can be interpreted as methods of apportionment. We consider several multiwinner rules and observe that some, but not all, of them induce apportionment methods 
	that are well established in the literature and in the actual practice of representation, be it proportional or non-proportional. For instance, we show that Proportional Approval Voting induces the \dhondt and that Monroe's rule induces the \hamilton. Our approach also yields apportionment methods implementing degressive proportionality. Furthermore, we consider properties of apportionment methods and exhibit multiwinner rules that induce apportionment methods satisfying these properties.
\end{abstract}

%

%

\section{Introduction}

In a parliamentary election, the candidates are traditionally organized in political parties and the election determines how many parliamentary seats each party is allocated. Under so-called ``closed party lists'' systems of representation, a voter is allowed to give her vote to one and only one party. In a sense, this forces the voter to approve all candidates from one party and no candidate from any other party. Counting the ballots under a closed-list system amounts to deciding how many seats each party gets based on the ``size'' of each party in the electorate, and is therefore formally an \textit{apportionment} problem, identical for instance to the problem of allocating seats to the states of a union based on the population figures.

Though widely used \citep{Farrell01}, closed-list systems have a number of drawbacks. For instance, it is a known feature of these systems that
candidates tend to campaign \emph{within} their parties (for being placed on a good position on the party list), rather than  campaigning for the citizens' votes. Closed-list systems thus favor party discipline, at the potential expense of alienating the political elites from the citizens
\citep[\eg see][]{personalRepresentation, audrey2015electoral, ames1995ElectoralStrategies,chang2005ElectoralIncentives}.

In an attempt to overcome these drawbacks, many countries use various forms of ``open-list'' systems, keeping the idea of party lists but giving some flexibility to the voters by allowing them to vote for specific candidates inside their chosen party list. 
In some (rare) cases, voters are given even more freedom. Under so-called ``panachage'' systems, sometimes used in Luxembourg and in France, voters can vote for candidates from different parties. Personalization is complete when a rule allows voters to vote directly for individual candidates, and the outcome of the election does not depend on how candidates are grouped into parties. Such is the case for some elections in Switzerland, were variants of multiwinner approval voting are used in several cantons \citep[see][]{laslier2016StrategicVoting}.
In a recent book, \citet{RePi16a} extensively examine the trend towards greater personalization, which they see as ``one of the key shifts in contemporary politics.''

This paper will take as object such fully open voting systems. In the literature they are called ``multiwinner'' rules because they elect a fixed number of candidates, for instance a whole parliament.
In a multiwinner election, we are given a set of voters who entertain preferences over a set of alternatives. Based on these preferences, the goal is to select a \emph{committee}, \ie a (fixed-size) subset of the alternatives. Preferences are usually specified either as rankings, i.e., linear orders over the set of all alternatives \citep[\eg][]{BramsFishburn02a,elk-fal-sko-sli:c:multiwinner-rules}, or as approval votes, i.e., yes/no assessments of the alternatives \citep[\eg][]{KBS06a}. We are particularly interested in the latter variant, where each voter can be thought of as specifying a subset of alternatives that she finds ``acceptable''.  

The decision scenario modeled by multiwinner elections---selecting a subset of objects from a potentially much larger pool of available objects---is ubiquitous, and includes picking players to form a sports team, selecting items to display in an online shop, choosing the board of directors of a company, etc. (for a more detailed discussion on the applications of multiwinner election rules, we refer the reader to the recent survey by \citet{FSST-trends}), 
but we are here chiefly interested in the political context. In this context, even if the voting rule allows personal candidacy and votes, most candidates do belong to political parties. It therefore makes sense to study, as a theoretical benchmark, the behavior of a multiwinner rule in the specific cases where voter preferences perfectly reflect some underlying party structure.
 
In order to do so, we imagine that the ballots cast are such that, when a voter votes for a candidate belonging to some party, she also votes for all the candidates of that party, and for no other candidates. As we have already mentioned, counting such ballots, and deciding how many candidates are elected from each party, is an apportionment problem. Any apportionment problem can be seen as a very simple approval voting instance: all voters approve all the candidates from their chosen party, and only those.
By contrast, in the context of approval-based multiwinner elections, voters are allowed to express more fine-grained preferences over individual candidates rather than over whole political parties.  
 
The present paper formally establishes and explores this analogy between multiwinner elections and apportionment problems. We show how an apportionment problem can be phrased as an instance of an approval-based multiwinner election, thereby rendering multiwinner rules applicable to the apportionment setting. As a result, every approval-based multiwinner rule induces a method of apportionment. Exploring this link between multiwinner rules and apportionment methods is interesting for at least two reasons. First, observing what kind of apportionment method a given multiwinner rule induces, yields new insights into the nature of the rule. Indeed, we argue that in order to understand multiwinner voting rules, one should first understand what they do on the important subdomain of apportionment.
Second, every multiwinner rule inducing a given apportionment method can be seen as an extension of the apportionment method to a more general setting where candidates have no party affiliations (or party affiliations are ignored in the election process).

After formally establishing the link between approval-based multiwinner rules and apportionment methods, we consider several known multiwinner rules and observe that some of them induce (and extend) apportionment methods that are well-established in the apportionment literature. For instance, \emph{Proportional Approval Voting (PAV)} induces the \dhondt (aka the Jefferson method) and \emph{Monroe's rule} induces the \hamilton (aka the Hamilton method). 
We also consider properties of apportionment methods related to proportionality (such as \emph{lower quota}), degressive proportionality (such as the \emph{Penrose condition}), or thresholds of representation, and exhibit multiwinner rules that induce apportionment methods satisfying these properties. Consequently, our work provides new tools for a normative comparison of various multiwinner voting rules, with respect to different forms of proportionality and non-proportionality.    



The paper is organized as follows.
\secref{sec:setting} formally introduces both the apportionment problem and the multiwinner election setting. 
\secref{sec:apportionmentConstruction}  shows how approval-based multiwinner rules can be employed as apportionment methods, and contains several results related to proportional representation.
\secref{sec:degressive}  is devoted to non-proportional representation in the form of degressive proportionality and thresholds, and \secref{sec:conclusion}  concludes.

\section{The Apportionment Problem and Approval-Based Multiwinner Elections}
\label{sec:setting}

In this section we provide the formal setting for the apportionment problem and for approval-based multiwinner elections. For a natural number $t \in \naturals=\{1,2,\ldots\}$, let $[t]$ denote the set $\{1, 2, \ldots, t\}$.

\subsection{Apportionment Methods}

In the apportionment setting, there is a finite set of voters and a finite set of $p$ parties $P_1, \ldots, P_p$. Every voter votes for exactly one party, and for each $i \in [p]$, we let $v_i$ denote the number of votes that party $P_i$ receives, \ie the number of voters who voted for $P_i$. The goal is to allocate~$h$ (parliamentary) \emph{seats} among the parties. Formally, an instance of the apportionment problem is given by a tuple $(v,h)$, where $v = (v_1,\ldots,v_p) \in \naturals^p$ is the vote distribution and $h \in \naturals$ is the number of seats to distribute. We use $v_{+}$ to denote the total number of votes, $v_{+} = \sum_{i=1}^p v_i$. Throughout this paper, we assume that $v_i>0$ for all $i \in [p]$ and $h>0$. An \emph{apportionment method} $M$ maps every instance $(v,h)$ to a nonempty set\footnote{Most apportionment methods allow for ties. In this paper we do not consider any specific tie-breaking rule but rather assume that an apportionment method can return several tied outcomes.} 
$M(v,h)$ of \emph{seat distributions}. A seat distribution is a vector $(x_1, \ldots, x_p) \in (\naturals \cup \{0\})^p$ with $\sum_{i=1}^p x_i = h$. Here, $x_i$ corresponds to the number of seats allocated to party $P_i$.

\newcommand{\move}[2]{{#1 \curvearrowright #2}}	
	
	In our proofs we often argue about seat distributions that result from a given seat distribution by taking away a single seat from a party and giving it to another party. For a seat distribution $x=(x_1,\ldots,x_p)$ and two parties $P_i$ and $P_j$ such that $x_i>0$, let $x_\move{i}{j}$ denote the seat distribution $x'$ with $x'_i = x_i-1$, $x'_j = x_j +1$, and $x'_\ell = x_\ell$ for all $\ell \notin \{i,j\}$. 

\subsubsection{Divisor Methods}

A rich and very well-studied class of apportionment methods is defined via divisor sequences. 


\begin{definition}[Divisor method]
	Let $d=(d(0), d(1), d(2), \ldots)$ be a sequence with $0 < d(j) \le d(j+1)$ for all $j \in \naturals \cup \{0\}$. The \emph{divisor method based on $d$} is the apportionment method that maps a given instance $(v,h)$ to the set of all seat allocations that can result from the following procedure:
	Start with the empty seat allocation $(0, \ldots, 0)$ and iteratively assign a seat to a party $P_i$ maximizing $\nicefrac{v_i}{d(s_i)}$, where $s_i$ is the number of seats that have already been allocated to party $P_i$.
\end{definition}



Divisor methods are often defined in a procedurally different, but mathematically equivalent way.%
\footnote{See \citet{BaYo82a}, Proposition 3.3. Divisor methods with $d(0)=0$ can also be defined. For such methods, which \citet{Puke14a} calls \emph{impervious}, the conventions $\frac{v_i}{0} = \infty$ and $\frac{v_i}{0} \ge \frac{v_j}{0} \Leftrightarrow v_i \ge v_j$ are used. Examples of impervious divisor methods are the methods due to Huntington and Hill, Adams, and Dean.\label{fn:impervious}} 
Two prominent divisor methods are the \dhondt and the \sainte, but several other divisor methods such as 
the \emph{Huntington-Hill method} (aka the \emph{method of equal proportions}), 
the \emph{Adams method} (aka the \emph{method of smallest divisors}), and
the \emph{Dean method}
are studied in the literature on fair representation. We refer the reader to the books of \citet{BaYo82a} and \citet{Puke14a} for an extensive overview on these methods, their history and their use in various countries. 

\begin{definition}[\dhondt]
	The \emph{\dhondt} 
	(aka the \emph{Jefferson method} or the \emph{Hagenbach-Bischoff method}) 
	is the divisor method based on the sequence $d_{\text{D'Hondt}}=(1,2,3,\ldots)$. Therefore, in each round, a seat is allocated to a party $P_i$ maximizing $\frac{v_i}{s_i+1}$, where $s_i$ is the number of seats that have already been allocated to party $P_i$.
\end{definition}

The \dhondt was first used in 1791 to apportion seats in the U.S. House of Representatives, and currently it is used as a legislative procedure in over 40 countries, often coupled with thresholds of representation (see \secref{sec:thresholds} on this last point). 

\begin{definition}[\sainte]
	The \emph{\sainte} 
	(aka the \emph{Webster method}, the \emph{Schepers method}, or the \emph{method of major fractions}) 
	is the divisor method based on the sequence $d_{\text{SL}}=(1,3,5,\ldots)$. Thus, in each round, a seat is allocated to a party $P_i$ maximizing $\frac{v_i}{2 s_i + 1}$, where $s_i$ is the number of seats that have already been allocated to party $P_i$.
\end{definition}

The \sainte was first adopted in 1842 for allocating seats in the United States House of Representatives, and currently it is used for parliamentary election in some countries (for instance, Latvia, New Zealand, Norway, and Sweden) and for several state-level legislatures in Germany.

The following example illustrates the two methods defined above.
\begin{example} \label{ex1}
Consider the instance with four parties $P_1, P_2, P_3, P_4$ and $100$ voters such that $(v_1, v_2, v_3, v_4) = (6, 7, 39, 48)$, and assume that there are $h = 10$ seats to be allocated. 
The outcomes of the \dhondt and the \sainte, which can be computed with the help of Table \ref{tab:dhondt-ex}, are $(0, 0, 4, 6)$ and $(1, 1, 4, 4)$, respectively.

\begin{table}[!htbp]
\centering
\begin{minipage}{0.45\linewidth}
\begin{align*}
\begin{array}{lrrrr} \toprule
  & P_1  & P_2  & P_3  & P_4  \\
  \midrule
\nicefrac{v_i}{1} & 6 & 7 & \mathbf{39} & \mathbf{48} \\
\nicefrac{v_i}{2} & 3 & 3.5 & \mathbf{19.5} & \mathbf{24} \\
\nicefrac{v_i}{3} & 2 & 2.33 & \mathbf{13} & \mathbf{16} \\
\nicefrac{v_i}{4} & 1.5 & 1.75 & \mathbf{9.75} & \mathbf{12} \\
\nicefrac{v_i}{5} & 1.2 & 1.4 & 7.8 & \mathbf{9.6} \\
\nicefrac{v_i}{6} & 1 & 1.17 & 6.5 & \mathbf{8.0} \\
\nicefrac{v_i}{7} & 0.86 & 1 & 5.57 & 6.86 \\
\bottomrule
\end{array}
\end{align*}
\end{minipage}
\begin{minipage}{0.45\linewidth}
\begin{align*}
\begin{array}{lrrrr} \toprule
  & P_1  & P_2  & P_3  & P_4  \\
  \midrule 
\nicefrac{v_i}{1} & \mathbf{6} & \mathbf{7} & \mathbf{39} & \mathbf{48} \\
\nicefrac{v_i}{3} & 2 & 2.33 & \mathbf{13} & \mathbf{16} \\
\nicefrac{v_i}{5} & 1.2 & 1.4 & \mathbf{7.8} & \mathbf{9.6} \\
\nicefrac{v_i}{7} & 0.86 & 1.0 & \mathbf{5.57} & \mathbf{6.86} \\
\nicefrac{v_i}{9} & 0.67 & 0.78 & 4.33 & 5.33 \\ 
\bottomrule
\end{array}
\end{align*}
\end{minipage}
\caption{Tables illustrating the computation of the \dhondt (left table) and \sainte (right table) in Example~\ref{ex1}. A column corresponding to party $P_i$ contains the ratios $\nicefrac{v_i}{d(0)}, \nicefrac{v_i}{d(1)}, \nicefrac{v_i}{d(2)}, ...$ with the sequences of denominators given by the respective divisor method (the ratios are rounded to two decimal places). The $h=10$ largest numbers in each table are printed in bold and correspond to the resulting seat allocation.}
\label{tab:dhondt-ex}
\end{table}
\end{example}


As we can see in the above example, different divisor methods might give different results for some instances of the apportionment problem.
In particular, the \dhondt slightly favors large parties over small ones in comparison to the \sainte~\citep[see, \eg][]{BaYo82a}.

\subsubsection{The \HAMILTON}

The \hamilton is the most well-known apportionment method that is not a divisor method. 
Recall that $v_+ = \sum_{i=1}^p v_i$ denotes the total number of votes.  

\begin{definition}[\Hamilton]
	The \emph{\hamilton}
	(aka the \emph{Hamilton method} or the \emph{Hare-Niemeyer method}) is defined via two steps. In the first step, each party $P_i$ is allocated $\lfloor \frac{v_ih}{v_+} \rfloor$ seats. In the second step, the remaining seats are distributed among the parties so that each party gets at most one of them. To do so, the parties are sorted according to the remainders $\frac{v_ih}{v_+} - \lfloor \frac{v_ih}{v_+} \rfloor$ and the remaining seats are allocated to the parties with the largest remainders.
\end{definition}

The \hamilton was first proposed by Alexander Hamilton in 1792 and it was used as a rule of distributing seats in the U.S. House of Representatives between 1852 and 1900. Currently, it is used for parliamentary elections in Russia, Ukraine, Tunisia, Namibia, and Hong Kong.

\begin{example}
Consider again the instance in Example~\ref{ex1}. According to the \hamilton, in the first step, party $P_1$ is allocated $\lfloor \frac{6 \cdot 10}{100} \rfloor = 0$ seats, $P_2$ gets $\lfloor \frac{7 \cdot 10}{100} \rfloor = 0$ seats, $P_3$ gets $\lfloor \frac{39 \cdot 10}{100} \rfloor = 3$ seats, and $P_4$ gets $\lfloor \frac{48 \cdot 10}{100} \rfloor = 4$ seats. There are three more seats to be distributed among the parties. The three highest fractional parts are $\frac{9}{10}$ for party $P_3$, $\frac{8}{10}$ for party~$P_4$ and  $\frac{7}{10}$ for party~$P_2$, thus these three parties are allocated one additional seat each. The resulting seat allocation is $(0,1,4,5)$.
\end{example}

\subsubsection{Properties of Apportionment Methods Related to Proportionality}

The literature on proportional representation has identified a number of desirable properties of apportionment methods \citep{BaYo82a,Puke14a}. 
Here, we focus on properties requiring that the proportion of seats in the resulting apportionment should reflect, as close as possible, the proportion of the votes cast for respective parties. 

\begin{definition}
An apportionment method \emph{respects lower quota} if, for every instance $(v,h)$, each party $P_i$ gets at least $\lfloor \frac{v_i h}{v_+} \rfloor$ seats. An apportionment method \emph{respects quota} if each party $P_i$ gets either $\lfloor \frac{v_ih}{v_+} \rfloor$ or $\lceil \frac{v_ih}{v_+} \rceil$ seats. 
\end{definition}

Clearly, any apportionment method respecting quota also respects lower quota. It is well known that the \hamilton respects quota, but no divisor method does. Moreover, the \dhondt is the only divisor method respecting lower quota \citep{BaYo82a}.

\subsection{Approval-Based Multiwinner Election Rules}\label{sec:multiwinnerOverview}

We now introduce the setting of approval-based multiwinner elections. We have a finite set $N = \{1, 2, \ldots, n\}$ of voters and a finite set $C = \{c_1, c_2, \ldots, c_m\}$ of candidates. Each voter expresses their preferences by approving a subset of candidates, and we want to select a committee consisting of exactly $k$ candidates. We will refer to $k$-element subsets of $C$ as \emph{size-$k$ committees}, and we let $A_i \subseteq C$ denote the set of candidates approved by voter $i \in N$.
Formally, an instance of the approval-based multiwinner election problem is given by a tuple $(A,k)$, where $A = (A_1, A_2, \ldots, A_n)$ is a \emph{preference profile} and $k$ is the desired committee size. An \emph{approval-based multiwinner election rule} (henceforth \emph{multiwinner rule}) $\calR$ is a function that maps every instance $(A,k)$ to a nonempty set\footnote{In order to accommodate tied outcomes, several committees might be winning.} $\calR(A,k)$ of size-$k$ committees. Every element of $\calR(A,k)$ is referred to as a \emph{winning committee}.

\subsubsection{Thiele Rules}

A remarkably general class of multiwinner election rules was proposed by Danish polymath Thorvald N. \citet{Thie95a}. A \emph{weight sequence} is an infinite sequence of real numbers $w = (w_1, w_2, \ldots)$.


\begin{definition}[Thiele rules]
Consider a weight sequence $w$, a committee $S \subseteq C$, and a voter $i \in N$ with approval set $A_i \subseteq C$. The \emph{satisfaction} of $i$ from $S$ given $w$ is defined as 
$u_i^w(S) = \sum_{j=1}^{|A_i \cap S|} w_j$, with $u_i^w(S) =0$ if $A_i \cap S = \emptyset$. 
Given an instance $(A,k)$ of the multiwinner election problem, the rule \emph{$w$-Thiele} selects all size-$k$ committees $S$ that maximize the total satisfaction $\sum_{i \in N} u_i^w(S)$.	
\end{definition}

Note that multiplying a weight sequence $w$ by a positive constant does not change the way in which a rule operates.
Several established multiwinner election rules can be described as Thiele rules. 

\begin{definition}
The \emph{Chamberlin--Courant rule} is $w_{\text{CC}}$-Thiele, where $w_{\text{CC}} = (1, 0, 0, \ldots)$. 
\end{definition}

The Chamberlin--Courant rule is usually defined in the context of multiwinner elections where voter preferences are given by \emph{ranked ballots}, and each voter derives satisfaction only from their most preferred member of the committee \citep{ChCo83a}. Our definition of the rule is a straightforward adaption to the approval setting: a voter is satisfied with a committee if and only if it contains at least one candidate that the voter approves of.


\begin{definition}
\emph{Proportional Approval Voting (PAV)} is $w_{\text{PAV}}$-Thiele, where $w_{\text{PAV}} = (1, \nicefrac{1}{2}, \nicefrac{1}{3}, \ldots)$.  
\end{definition}

Though sometimes attributed to Forest Simmons, PAV was already proposed and discussed by Thiele in the 19th century \citep{Thie95a,Jans16a}.\footnote{We are grateful to Xavier Mora and Svante Janson for pointing this out to us.} 
According to PAV, each voter cares about the whole committee, but the marginal gain of satisfaction of an already satisfied voter from an additional approved committee member is lower than the gain of a less satisfied voter. The reason for using the particular weight sequence $w_{\text{PAV}}$ is not obvious. \citet{ABC+17a} and \citet{SFF+17a} provide compelling arguments by showing that $w_{\text{PAV}}$ is the unique weight sequence $w$ such that $w$-Thiele satisfies certain axiomatic properties. \thmref{thm:pavApportionment2} in the present paper can be viewed as an additional---though related---argument in favor of the weight sequence~$w_{\text{PAV}}$.

\begin{definition}
The \emph{top-$k$ rule} is $w_{\text{top-$k$}}$-Thiele, where $w_{\text{top-$k$}} = (1, 1, 1, \ldots)$.
\end{definition}

According to the top-$k$ rule, the winning committee contains the $k$ candidates that have been approved by the greatest number of voters.
This very natural rule is called \emph{Simple Approval} by \cite{pavVoting}.

The following example illustrates the Thiele rules defined above.

\begin{example}\label{ex:owaBasedRules}
Let $k=3$ and consider the following preference profile with 12 voters:
\begin{align*}
& A_1 = \{c_1\} \quad & & A_5 = \{c_1, c_4, c_5, c_6\} \quad & & A_9 = \{c_2\} \\
& A_2 = \{c_1, c_3, c_5 \} \quad & & A_6 = \{c_1, c_4, c_5, c_6\} \quad & & A_{10} = \{c_3, c_5\} \\
& A_3 = \{c_1, c_5, c_6 \} \quad & & A_7 = \{c_1, c_4\} \quad & & A_{11} = \{c_3\} \\
& A_4 = \{c_1, c_5, c_6 \} \quad & & A_8 = \{c_2, c_4, c_6\} \quad & & A_{12} = \{c_3\}
\end{align*}
According to the Chamberlin--Courant rule, the unique winning committee is $\{c_1, c_2, c_3\}$. Indeed, for this committee each voter approves of at least one committee member (voters $1$--$7$ approve of~$c_1$, voters $8$ and $9$ approve of~$c_2$, and voters $10$--$12$ approve of $c_3$), thus each voter gets satisfaction equal to 1 from this committee. Clearly, this is the highest satisfaction that the voters can get from a committee. The reader might check that $\{c_1, c_2, c_3\}$ is the only committee for which each voter has at least one approved committee member. 

According to Proportional Approval Voting, $\{c_1, c_3, c_6\}$ is the unique winning committee. Let us compute the PAV satisfaction of voters from this committee. Voters 1, 7, 8, 10--12 approve of a single committee member, thus the satisfaction of these voters is equal to 1. Voter 9 does not approve any committee member, thus her satisfaction is equal to 0. The remaining 5 candidates approve of two committee members, so their satisfaction from the committee is equal to $1 + \frac{1}{2}$. Thus, the total satisfaction of the voters from $\{c_1, c_3, c_6\}$ is equal to $6 + 5\cdot \frac{3}{2} = 13.5$ and it can be checked that the satisfaction of voters from all other committees is lower.

The top-$k$ rule, on the other hand, selects $\{c_1, c_5, c_6\}$ as the unique winning committee. Indeed, $c_1$, $c_5$, and $c_6$ are approved by 7, 6, and 5 voters, respectively, and any other candidate is approved by less than 5 voters.
\end{example}

Other appealing Thiele rules include the \emph{$t$-best Thiele rule}, which is defined by the weight sequence $w = (1, \ldots, 1, 0, 0 \ldots)$ with $t$ ones followed by zeros, and the \emph{$t$-median Thiele rule}, which is defined by the weight sequence $w = (0, \ldots, 0, 1, 0, 0,\ldots)$, where the $1$ appears at position~$t$ \citep{owaWinner}.

\subsubsection{Sequential Thiele Rules}

Another interesting class of multiwinner rules, sometimes referred to as \emph{Reweighted Approval Voting (RAV)}, consists of \emph{sequential} variants of Thiele rules.

\begin{definition}
The rule \emph{sequential $w$-Thiele} (aka $w$-RAV) selects all committees that can result from the following procedure. Starting with the empty committee ($S = \emptyset$), in $k$ consecutive steps add to the committee $S$ a candidate $c$ that maximizes $\sum_{i \in N} u_i^w(S \cup \{c\}) - u_i^w(S)$. 
\end{definition}


Just like Thiele rules, sequential Thiele rules have already been considered by \citet{Thie95a}.
\emph{Sequential Proportional Approval Voting} (i.e., sequential $w_{\text{PAV}}$-Thiele) 
was used for a short period in Sweden during the early 1900s \citep{Jans16a}.
 
Interestingly, a sequential Thiele rule approximates the original Thiele rule whenever the weight sequence $w$ is non-increasing \citep{owaWinner}. 
Sequential Thiele rules are appealing alternatives to their original variants for a number of reasons. 
	First, they are computationally tractable, while finding winning committees for (non-sequential) Thiele rules is often $\np$-hard. 
	Second, sequential Thiele rules satisfy properties that are violated by non-sequential Thiele rules;
for instance, they (trivially) satisfy committee monotonicity while some of the non-sequential versions do not \citep[see, \eg][]{elk-fal-sko-sli:c:multiwinner-rules}.
	Third, sequential methods are easier to describe to non-experts, as compared to rules formulated as non-trivial combinatorial optimization problems. 

\begin{example}
Consider the instance from Example~\ref{ex:owaBasedRules}. Sequential Proportional Approval Voting selects candidate $c_1$ in the first iteration. In order to decide which candidate is selected in the second iteration, one must compare all remaining candidates with respect to their contribution to the total satisfaction of the voters. For instance, by selecting $c_3$ the satisfaction of voters will increase by $\frac{1}{2} + 3$. One can check that there are two optimal choices in the second step: selecting $c_3$ or selecting~$c_5$. If $c_3$ is selected in the second step, then our procedure will select $c_6$ in the third step (selecting $c_6$ increases the satisfaction of voters by $1 + 4\cdot \frac{1}{2}$). If $c_5$ is selected in the second step, then $c_3$ must be chosen in the third step (by choosing $c_3$ we increase the satisfaction of voters by $\frac{1}{3} + 5\cdot \frac{1}{2}$). Consequently, there are two winning committees according to Sequential Proportional Approval Voting: $\{c_1, c_3, c_5\}$ and $\{c_1, c_3, c_6\}$.
\end{example}

\subsubsection{Monroe's Rule}
\label{sec:monroedef}

The optimization problem underlying the Chamberlin--Courant rule can be thought of in terms of maximizing representation: every voter is assigned to a single candidate in the committee and this ``representative'' completely determines the satisfaction that the voter derives from the committee. The rule proposed by \citet{monroeElection} is based on the same idea; however, Monroe requires each candidate to represent the same number of voters. For the sake of simplicity, when considering the Monroe rule we assume that the number of voters $n$ is divisible by the size of the committee $k$. 

\begin{definition}
	Assume that $k$ divides $n$ and consider a size-$k$ committee $S$. A \emph{balanced allocation} of the voters to the candidates in $S$ is a function $\tau_S: N \to S$ such that $|\tau_S^{-1}(c)| = \frac{n}{k}$ for all $c \in S$. The satisfaction  $u_i(\tau_S)$ of a voter $i$ from $\tau_S$ is equal to one if $i$ approves of $\tau_S(i)$, and zero otherwise. The total satisfaction of voters provided by $S$, denoted $u(S)$, is defined as the satisfaction from the best balanced allocation, \ie $u(S) = \max_{\tau_S} \sum_{i \in N} u_i(\tau_S)$. The \emph{Monroe rule} selects all committees~$S$ maximizing $u(S)$.
\end{definition}

\begin{example}
Consider again the instance from Example~\ref{ex:owaBasedRules} and consider committee $\{c_1, c_2, c_3\}$, which is a winning committee under the Chamberlin--Courant rule. Even though every voter has an approved committee member (voters $1$--$7$ approve of~$c_1$, voters $8$ and $9$ approve of~$c_2$, and voters $10$--$12$ approve of $c_3$), according to Monroe this committee has a total satisfaction of 10 only:
since only two voters approve of $c_2$, there is no balanced allocation $\tau_{\{c_1, c_2, c_3\}}$ mapping every voter to an approved candidate in $\{c_1, c_2, c_3\}$. 
%
%
The Monroe rule selects $\{c_1, c_4, c_3\}$ as the unique winning committee. For this committee, the optimal balanced allocation maps voters $1$--$4$ to their approved candidate $c_1$, voters $5$--$8$ to their approved candidate~$c_4$, and voters~$9$--$12$ to~$c_3$ (only three of those four voters approve of~$c_3$). Thus, the winning committee has a total satisfaction of $11$.  
\end{example}

\subsubsection{\phrag's Rules}

In the late 19th century, the Swedish mathematician Lars Edvard \phrag proposed several methods for selecting committees based on approval votes \citep{Phra94a,Phra95a,Phra96a,Phra99a}. Here, we formulate two particularly interesting variants \citep[see][]{Jans16a,BFJL16a}. 

The motivation behind \phrag's rules is to find a committee whose ``support'' is distributed as evenly as possible among the electorate. For every candidate in the committee, one unit of ``load'' needs to be distributed among the voters approving this candidate. The goal is to find a committee of size $k$ for which the maximal load of a voter (or the variance of load distribution) is as small as possible. 

\begin{definition}
	Consider an instance $(A,k)$ of the multiwinner election problem. A load distribution for $(A,k)$ is a matrix $L = (\ell_{i,c})_{i \in N, c \in C} \in \reals_{\ge 0}^{|N|\times |C|}$ such that 
	\begin{enumerate}
		\item[(i)] $\sum_{i\in N} \sum_{c\in C}  \ell_{i,c} = k$,
		\item[(ii)] $\sum_{i\in N} \ell_{i,c} \in \{0,1\}$ for all $c\in C$, and 
		\item[(iii)] $\ell_{i,c} = 0$ if $c\notin A_i$.
	\end{enumerate}
Every load distribution $L$ corresponds to a size-$k$ committee $S_L=\{c\in C: \sum_{i\in N} \ell_{i,c}= 1\}$.
\end{definition}

The multiwinner election rules \maxP and \varP are defined as optimization problems over the set of all load distributions.\footnote{Minimizing the variance of voter loads is equivalent to minimizing $\sum_{i\in N} (\sum_{c\in C} \ell_{i,c})^2$.}

\begin{definition}
	The rule \emph{\maxP} maps an instance $(A,k)$ to the set of size-$k$ committees corresponding to load distributions $L$ minimizing $\max_{i\in N} \sum_{c\in C} \ell_{i,c}$.
	The rule {\varP} maps an instance $(A,k)$ to the set of size-$k$ committees corresponding to load distributions $L$ minimizing $\sum_{i\in N} (\sum_{c\in C} \ell_{i,c})^2$.
\end{definition}

\begin{example}\label{ex:phragmen}
	Let $k=3$ and consider the following preference profile with 5 voters:
	\begin{align*}
	& A_1 = \{c_1\}  \quad & &   A_3 = \{c_2, c_3\} \quad & &  A_5 = \{c_4\} \\
	& A_2 = \{c_2 \} \quad & &   A_4 = \{c_1, c_2, c_3\}  & &
	\end{align*}
	It can be checked that \maxP selects the committee $\{c_1,c_2,c_3\}$ and that \varP selects the committee $\{c_1,c_2,c_4\}$. Optimal load distributions corresponding to these committees are illustrated in \figref{fig:phragmen}.
	Load distributions minimizing the maximal voter load (like the one illustrated on the left in \figref{fig:phragmen}) satisfy 
	$\max_{i\in N} \sum_{c\in C} \ell_{i,c} = \frac{3}{4}$ and 
	$\sum_{i\in N} (\sum_{c\in C} \ell_{i,c})^2 = 4 \cdot \left(\frac{3}{4}\right)^2 = \frac{9}{4}$, 
	and the load distribution minimizing the variance of voter loads (illustrated on the right in \figref{fig:phragmen}) satisfies 
	$\max_{i\in N} \sum_{c\in C} \ell_{i,c} = 1$ and 
	$\sum_{i\in N} (\sum_{c\in C} \ell_{i,c})^2 = 4 \cdot \left(\frac{1}{2}\right)^2 + 1^2 = 2$. 
\end{example}

\begin{figure}
	\centering
\begin{tikzpicture}[yscale=.45,xscale=.4]
	\draw (0,7) rectangle node {$c_1$} +(6,1); 
	\draw (0,6) rectangle node {$c_2$} +(6,1);
	\draw (0,5) rectangle node {$c_3$} +(6,1);
	\draw (0,4) rectangle node {$c_1$} +(2,1); \draw (2,4) rectangle node {$c_2$} +(2,1); \draw (4,4) rectangle node {$c_3$} +(2,1);
	
	\draw (-0.5,7.5) node[left]{$A_1 = \{c_1\}$};
	\draw (-0.5,6.5) node[left]{$A_2 = \{c_2 \}$};
	\draw (-0.5,5.5) node[left]{$A_3 = \{c_2, c_3\}$};
	\draw (-0.5,4.5) node[left]{$A_4 = \{c_1, c_2, c_3\}$};
	\draw (-0.5,3.5) node[left]{$A_5 = \{c_4\}$};
	
	\draw[->] (0,3)--(7,3);
	\draw (0,3) -- (0,2.8)
	node[anchor=north] {$0$};
	\draw (2,3) -- (2,2.8)
	node[anchor=north] {$\frac{1}{4}$};
	\draw (4,3) -- (4,2.8)
	node[anchor=north] {$\frac{1}{2}$};
	\draw (6,3) -- (6,2.8)
	node[anchor=north] {$\frac{3}{4}$};
\end{tikzpicture}
\hfill
\begin{tikzpicture}[yscale=.5,xscale=.6]
	\draw (0,7) rectangle node {$c_1$} +(4,1); 
	\draw (0,6) rectangle node {$c_2$} +(4,1);
	\draw (0,5) rectangle node {$c_2$} +(4,1);
	\draw (0,4) rectangle node {$c_1$} +(4,1); 
	\draw (0,3) rectangle node {$c_4$} +(8,1); 
	
	\draw (-0.5,7.5) node[left]{$A_1 = \{c_1\}$};
	\draw (-0.5,6.5) node[left]{$A_2 = \{c_2 \}$};
	\draw (-0.5,5.5) node[left]{$A_3 = \{c_2, c_3\}$};
	\draw (-0.5,4.5) node[left]{$A_4 = \{c_1, c_2, c_3\}$};
	\draw (-0.5,3.5) node[left]{$A_5 = \{c_4\}$};
	
	\draw[->] (0,3)--(9,3);
	\draw (0,3) -- (0,2.8)
	node[anchor=north] {$0$};
	\draw (4,3) -- (4,2.8)
	node[anchor=north] {$\frac{1}{2}$};
	\draw (8,3) -- (8,2.8)
	node[anchor=north] {$1$};
\end{tikzpicture}
\caption{Illustration of Example \ref{ex:phragmen}. The diagram on the left illustrates a load distribution minimizing the maximal voter load, and the diagram on the right illustrates the unique load distribution minimizing the variance of voter loads.}
\label{fig:phragmen}
\end{figure}
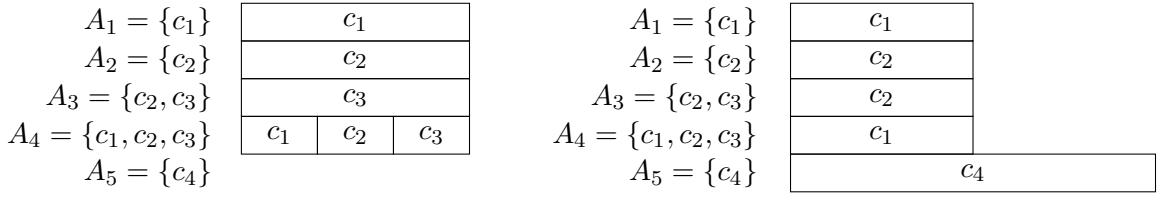

\section{Apportionment Via Multiwinner Election Rules}
\label{sec:apportionmentConstruction}

In this section we demonstrate how approval-based multiwinner rules can be employed as apportionment methods. For a given instance $(v,h)$ of the apportionment problem, this procedure involves three steps:
\begin{enumerate}
	\item Translate $(v,h)$ into an instance $(A,k)$ of the multiwinner problem. 
	\item Apply the multiwinner rule $\calR$ to $(A,k)$.
	\item Translate committee(s) in $\calR(A,k)$ into seat distribution(s) for $(v,h)$.
\end{enumerate} 

\noindent We now describe each step in detail. 
\paragraph{Step 1.}
 Given an instance $(v,h)$ of the apportionment problem, we construct an instance $(A,k)$ of the multiwinner problem as follows. For each $i \in [p]$, we introduce a set $C_i$ consisting of $h$ candidates, and a set $N_i$ consisting of $v_i$ voters. Each voter in $N_i$ approves of all candidates in $C_i$ (and of no other candidates).
Furthermore, we define $C=\bigcup_{i=1}^{p} C_i$, $N=\bigcup_{i=1}^{p} N_i$, and $k=h$.
Intuitively, $C_i$ is the set of members of party $P_i$ and $N_i$ is the set of voters who voted for party $P_i$.

\paragraph{Step 2.}
We can now apply multiwinner rule $\calR$ to $(A,k)$ in order to find the set $\calR(A,k)$ of winning committees.

\paragraph{Step 3.}
For every winning committee $S \in \calR(A,k)$, we can extract a seat distribution $x$ for the instance $(v,h)$ in the following way: the number $x_i$ of seats of a party $P_i$ is given by the number of candidates from $C_i$ in the committee $S$, i.e., $x_i = |C_i \cap S|$.

\medskip
The next example illustrates this three-step procedure. 

\begin{example}[Sequential PAV as an apportionment method]
Consider the instance $(v,h) = ((20,40,30,10),10)$ of the apportionment problem. We construct an instance of the multiwinner election setting with $40$ candidates; each party $P_i$ has a set $C_i$ consisting of $10$ candidates. Further, there are $\sum_{i=1}^4 v_i = 100$ voters: $20$ voters approve the ten candidates in $C_1$, $40$ voters approve the ten candidates in $C_2$, $30$ voters the ten candidates in $C_3$, and $10$ voters the ten candidates in~$C_4$. Applied to this instance, Sequential PAV selects a winning committee consisting of 
two candidates from $C_1$,
four candidates from $C_2$,
three candidates from $C_3$, and
one candidate from~$C_4$.
Thus, the resulting seat allocation is $(x_1, x_2, x_3, x_4) = (2, 4, 3, 1)$.
\end{example}

For a given multiwinner rule $\calR$, we let $M_\calR$ denote the apportionment method defined by steps~1 to~3. 
We say that a multiwinner election rule $\calR$ \emph{satisfies (lower) quota} if the corresponding apportionment method $M_\calR$ satisfies the respective property.

\subsection{Apportionment Methods Induced by Thiele Rules}

In this section, we consider apportionment methods induced by Thiele rules. Fix a weight sequence $w$ and let $\calR$ denote the rule $w$-Thiele. For every instance $(v,h)$ of the apportionment setting,  $M_{\calR}(v,h)$ contains all seat distributions $x$ maximizing the total voter satisfaction $u(x)$, which is given by 
	\begin{equation} \label{eq:total}
		u(x) = \sum_{i \in [p]} v_i u_i^w(x) =  \sum_{i \in [p]} v_i \left( \sum_{j=1}^{x_i} w_j \right) \text.
	\end{equation}
	Here, $u_i^w(x) = \sum_{j=1}^{x_i} w_j$ denotes the satisfaction of a voter from party $P_i$. 

	We let $\hat{u}(i,s)$ denote the marginal increase in $u(\cdot)$ when assigning the $s$-th seat to party $P_i$. Due to (\ref{eq:total}), we have $\hat{u}(i,s) = v_i w_s$.
	Taking away a seat from party $P_i$ and giving it to party $P_j$ results in a change of total voter satisfaction of
	\[ u(x_\move{i}{j}) - u(x)  = \hat u(j,x_j+1) - \hat u(i,x_i) = 
								  v_j w_{x_j+1} - v_i w_{x_i} \text. \]
	Optimality of $x \in M_{\calR}(v,h)$ implies that $u(x_\move{i}{j}) - u(x) \le 0$ for all $i,j \in [p]$ with $x_i>0$.

\medskip

Whenever the weight vector $w$ is non-increasing, $w$-Thiele induces the same apportionment method as its sequential variant.%
\footnote{Thiele rules with non-increasing weight sequences are very natural, especially when viewed in the context of apportionment. Thiele rules with \emph{increasing} sequences induce apportionment methods that allocate all seats to the single party that receives the most votes. The same seat distribution can be also obtained by using the Thiele rule with a constant weight vector (see Proposition~\ref{prop:topKapportionment}).} 

\begin{proposition}\label{prop:sequential}
	Let $w$ be a weight sequence with $w_j \ge w_{j+1} \ge 0$ for all $j \in \naturals$. The apportionment method induced by $w$-Thiele coincides with the apportionment method induced by sequential $w$-Thiele. 
\end{proposition}

\begin{proof}
	Fix a weight sequence $w = (w_1, w_2, \ldots)$ satisfying $w_j \ge w_{j+1} \ge 0$ for all $j \in \naturals$.
	Let~$\calR$ be $w$-Thiele and let $\calR'$ be sequential $w$-Thiele. We show that the apportionment methods induced by $\calR$ and $\calR'$ coincide, \ie $M_{\calR'} = M_{\calR}$.
	
	Consider an instance $(v,h)$ of the apportionment setting. Then $M_{\calR}(v,h)$ contains all seat distributions $x$ maximizing the total voter satisfaction $u(x)$, which is given by (\ref{eq:total}). 
	
	Now consider the apportionment method $M_{\calR'}$ induced by sequential $w$-Thiele. This method starts with the empty seat allocation $x^{(0)}= (0, \ldots, 0)$ and iteratively assigns a seat to a party $P_i$ such that the \emph{marginal increase} in total voter satisfaction is maximized.  
	Recall that $\hat{u}(i,s)$ denotes the marginal increase in $u(\cdot)$ when assigning the $s$-th seat to party $P_i$. Due to~(\ref{eq:total}), this quantity is independent of $x_j$ for $j \ne i$ and equals  $\hat{u}(i,s) = v_i w_s$.
	For $\ell=1, \ldots, h$, method $M_{\calR'}$ iteratively chooses a party $P_{i^*}$ with 
	$i^* \in \arg \max_{i \in [p]} \hat{u}(i,x^{(\ell-1)}_i + 1)$ and sets 
	$x^{(\ell)} = (x^{(\ell-1)}_1, \ldots, x^{(\ell-1)}_{i^*} + 1, \ldots, x^{(\ell-1)}_p)$. 
	
	For every seat allocation $x$, the total voter satisfaction $u(x)$ can be written as $u(x) = \sum_{i \in [p]} \sum_{s=1}^{x_i} \hat{u}(i,s)$. In particular, $u(x)$ is independent of the order in which seats are allocated to parties. Furthermore, our assumption $w_1 \ge w_2 \ge \ldots$ implies that the marginal satisfaction $\hat{u}(i,s)$ is monotonically decreasing in $s$ for every party $P_i$. Therefore, any seat distribution maximizing $u(\cdot)$ (\ie any $x \in M_\calR(v,h)$) can be iteratively constructed by applying method $M_{\calR'}$. 
\end{proof}

An immediate corollary is that Thiele rules with a non-increasing weight sequence, such as PAV, can be computed efficiently in the apportionment setting.

 Further, we can use \propref{prop:sequential} to show that every Thiele rule with a positive non-increasing weight sequence induces a divisor method. For a weight sequence $w = (w_1, w_2, \ldots)$ with $w_j>0$ for all $j \in \naturals$, let $d_w = (d_w(0), d_w(1), \ldots)$ be the sequence defined by $d_w(s) = \nicefrac{1}{w_{s+1}}$ for all $s \in \naturals \cup \{0\}$. For example, $w_\text{PAV}=(1,\nicefrac{1}{2}, \nicefrac{1}{3}, \ldots)$ yields $d_{w_\text{PAV}} = (1,2,3, \ldots)$.

\begin{theorem}\label{thm:owa-divisor}
	Let $w$ be a weight sequence with $w_j \ge w_{j+1} > 0$ for all $j \in \naturals$. 
	The apportionment method induced by $w$-Thiele coincides with the divisor method based on 
	$d_w=(\nicefrac{1}{w_1}, \nicefrac{1}{w_2}, \nicefrac{1}{w_3}, \ldots)$.
\end{theorem}

\begin{proof}
	Fix a weight sequence $w$ with $w_j \ge w_{j+1} > 0$ for all $j \in \naturals$. Let $\calR(w)$ denote the rule $w$-Thiele, and let $\calR'(w)$ denote sequential $w$-Thiele. It follows from \propref{prop:sequential} that $M_{\calR(w)}$ coincides with $M_{\calR'(w)}$.
	Furthermore let $M(d_w)$ denote the divisor method based on~$d_w$. 
	We are going to show that $M_{\calR'(w)}$ coincides with $M(d_w)$.
	
	Both $M(d_w)$ and $M_{\calR'(w)}$ work iteratively. 
	At each iteration, $M(d_w)$ assigns a seat to a party~$P_i$ maximizing $\frac{v_i}{d_w(s_i)}$, where $v_i$ is the number of votes for party $P_i$ and $s_i$ is the number of seats allocated to party $P_i$ in previous iterations. 
	Method $M_{\calR'(w)}$, on the other hand, assigns a seat to a party maximizing the marginal increase in total voter satisfaction. The marginal increase in total voter satisfaction when giving an additional seat to party $P_i$ equals $\hat{u}(i,s_i + 1) = v_i w_{s_i+1}$.

	Using $d_w(s) = \frac{1}{w_{s+1}}$, we can observe that the quantities $\left(\frac{v_i}{d_w(s_i)}\right)_{i \in [p]}$ used in divisor method $M(d_w)$ exactly coincide with the quantities $(v_i w_{s_i+1})_{i \in [p]}$ used in the method $M_{\calR'(w)}$. Thus, both methods assign seats in exactly the same way. 
\end{proof}

\subsection{Proportional Approval Voting as an Extension of the D'Hondt Method}

A particularly interesting consequence of \thmref{thm:owa-divisor} is that Proportional Approval Voting, the Thiele rule with $w_\text{PAV}=(1,\nicefrac{1}{2}, \nicefrac{1}{3}, \ldots)$, induces the \dhondt.

\begin{corollary} \label{cor:pav}
	The apportionment method induced by PAV coincides with the \dhondt.
\end{corollary}

The observation that PAV reduces to the \dhondt in the party-list setting occasionally occurs (without proof) in the literature \citep[\eg see][]{Plaz04a,Pere16a,SFF16a,BKP17a}.\footnote{In fact, already \citet{Thie95a} stated this equivalence, but only in the special case when perfect proportionality is possible; see the survey by \citet{Jans16a}.} \thmref{thm:owa-divisor} shows that this is just one special case of the general relationship between Thiele rules and divisor methods \citep[see also][Remark 11.4]{Jans16a}. 

Since the \dhondt satisfies lower quota, the same holds for PAV (and Sequential PAV) in the apportionment setting.

\begin{corollary}
	PAV and Sequential PAV satisfy lower quota. 
\end{corollary}

The fact that PAV satisfies lower quota can also be established by observing that every multiwinner rule satisfying \emph{extended justified representation} \citep{ABC+17a} or \emph{proportional justified representation} \citep{SFF+17a} also satisfies lower quota. See Appendix \ref{sec:jr} for details.

Further, we show that PAV is the \emph{only} Thiele rule satisfying lower quota.%
\footnote{This result also follows from \thmref{thm:owa-divisor} together with the characterization of the \dhondt as the only divisor method satisfying lower quota \citep[][Proposition 6.4]{BaYo82a}. 
We give a direct proof for completeness.}

\begin{theorem}\label{thm:pavApportionment2}
PAV is the only Thiele rule satisfying lower quota.
\end{theorem}

\begin{proof}
Let $\calR$ be a Thiele rule such that $M_\calR$ satisfies lower quota. We will show that $\calR$ is based on a weight sequence $w$ with $w_j = \nicefrac{w_1}{j}$ for all $j \in \naturals$, from which we will infer that $M_\calR$ is equivalent to PAV. 

Fix $j \in \naturals$. We will show $w_{j} = \nicefrac{w_1}{j}$ in two steps. 

\paragraph{Step 1.}
Given a natural number $Z>1$, we define an instance $(v,h)$ of the apportionment problem as follows. There are $Z+1$ parties and the vote distribution is given by $v = (Z, jZ-1, \ldots, jZ-1)$. That is, party $P_1$ gets $Z$ votes and all other parties get $jZ-1$ votes. Thus, the total number of votes is $v_+ = Z + Z(jZ - 1) = jZ^2$. Furthermore, we set $h = \frac{v_+}{Z} = jZ$. 

Consider a seat allocation $x \in M_\calR(v,h)$. Since $M_\calR$ satisfies lower quota, we have 
$x_1 \ge \left\lfloor \frac{Z \cdot jZ}{Z^2j}\right\rfloor = 1$ and 
$x_\ell \ge \left\lfloor \frac{(jZ-1) \cdot jZ}{Z^2j}\right\rfloor = \left\lfloor \frac{(jZ-1)}{z} \right\rfloor = j-1$ for all $\ell \ne 1$. Thus, $P_1$ is allocated at least one seat, and each of the other $Z$ parties is allocated at least $j-1$ seats. From the pigeonhole principle we infer that at least one of the parties $P_2, \ldots, P_{Z+1}$ gets \emph{exactly} $j-1$ seats. Let $P_\ell$ be such a party and consider the seat allocation $x_\move{1}{\ell}$. Since $x$ maximizes $u(\cdot)$, we have
\[ 0 \ge u(x_\move{1}{\ell}) - u(x) = \hat u(\ell,j) - \hat u(1,1) = (jZ-1) w_j - z w_1 \text. \]
It follows that $w_{j} \leq \frac{w_1}{j} \cdot \frac{Z}{Z - 1}$. Since this inequality has to hold for all natural numbers $Z>1$, we infer that $w_{j} \leq \frac{w_1}{j}$.

\paragraph{Step 2.}
Next, we construct another instance $(v',h')$ of the apportionment problem, again parameterized by a natural number $Z>1$. There are $Z+1$ parties and the vote distribution $v'$ is given by $v' = (jZ, Z-1, \ldots, Z-1)$. Thus, $v'_+ = Z^2 - Z + jZ$. The number of seats is given by $h' = \frac{v'_+}{Z} = Z + j - 1$. 

Consider a seat allocation $x \in M_\calR(v',h')$. Since $M_\calR$ satisfies lower quota, we infer that 
$x_1 \ge \left\lfloor \frac{jZ  (Z+j-1)}{Z(Z+j-1)} \right\rfloor = j$. The pigeonhole principle implies that at least one of parties $P_2, \ldots, P_{Z+1}$ gets no seat. Let $P_\ell$ be such a party and consider $x_\move{1}{\ell}$. We have 
\[ 0 \ge u(x_\move{1}{\ell}) - u(x) = \hat u(\ell,1) - \hat u(1,j) = (Z-1)w_1 - jZ \cdot w_j \text, \] 
and thus $w_{j} \ge \frac{w_1}{j} \cdot \frac{Z-1}{Z}$. Again, this inequality holds for all $Z$, and thus $w_{j} \geq \frac{w_1}{j}$.

\medskip
We have therefore shown that $w_{j} = \frac{w_1}{j}$ for all $j \in \naturals$. It follows that $w = c \cdot w_{\text{PAV}}$ for some constant $c>0$. As a consequence, $\calR$ is equivalent to PAV.
\end{proof}

\thmref{thm:pavApportionment2} characterizes PAV as the only Thiele rule respecting lower quota. Since PAV does not respect (exact) quota, it follows that no Thiele rule respects quota. Due to \propref{prop:sequential}, the same is true for sequential Thiele rules.

The load-balancing rule \maxP also induces the \dhondt. Indeed, \citeauthor{Phra95a} formulated his rule as a generalization of the \dhondt \citep{MoOl15a,Jans16a}. 

\begin{theorem} \label{thm:maxP}
	The apportionment method induced by \maxP coincides with the \dhondt.
\end{theorem}

\begin{proof}
In the apportionment setting, optimal load distributions have a very simple structure. Given a seat distribution $x$, it is clearly optimal to distribute the load of $x_i$ ($1$ for each seat that is allocated to party $P_i$) uniformly among the $v_i$ voters of the party. Therefore, the maximal voter load for~$x$ is given by $\max_{i \in [p]} \frac{x_i}{v_i}$. By definition, the \dhondt selects seat distributions minimizing this quantity \citep[see][Proposition~3.10]{BaYo82a}. 
\end{proof}

There is also a \emph{sequential} variant of \phrag's rule.\footnote{In fact, the sequential variant is the main method that \phrag proposed to be used in actual elections; see the survey by \citet{Jans16a} for details.}
It is straightforward to verify that, in the apportionment setting, this variant coincides with \maxP and thus also induces the \dhondt.  

\subsection{Monroe's Rule as an Extension of the \HAMILTON}

We now turn to Monroe's multiwinner rule. It turns out that it induces the \hamilton. 

\begin{theorem}\label{thm:monroeApportionment}
Assume that the number of seats divides the total number of voters. Then, Monroe's rule induces the \hamilton.
\end{theorem}

\begin{proof}
Let $\calR$ denote Monroe's rule. Recall that Monroe's rule assigns voters to candidates in a balanced way, so as to maximize the total voter satisfaction. In the apportionment setting, a voter in $N_i$ is satisfied if and only if she is assigned to a candidate in $C_i$. The apportionment method $M_\calR$ selects seat distributions maximizing the total voter satisfaction. 

For a seat allocation $x$, let $u(x)$ denote the total voter satisfaction provided by $x$. We can write 
$u(x) = \sum_{i \in [p]} u(i,x)$, where $u(i,x)$ is the total satisfaction that voters in $N_i$ derive from $x$. For a given instance $(v,h)$ of the apportionment problem, $u(i,x)$ can be expressed as
\[
u(i,x) = \begin{cases}
			v_i 			  &\text{if $x_i \ge v_i \frac{h}{v_+}$,} \\
			x_i \frac{v_+}{h} &\text{if $x_i < v_i \frac{h}{v_+}$.}
		 \end{cases}
\]
Since the case $u(i,x) = v_i$ occurs if and only if $v_i \le x_i \frac{v_+}{h}$, we have $u(i,x) = \min(v_i,x_i \frac{v_+}{h})$ and
\begin{equation*}\label{eq:monroe-sat}
	u(x) = \sum_{i \in [p]} u(i,x) = \sum_{i \in [p]} \min \left(v_i,x_i \frac{v_+}{h} \right) \text.
\end{equation*}

We show that $M_\calR$ coincides with the \hamilton for all instances $(v,h)$ such that $v_+$ divides $h$. Fix such an instance $(v,h)$ and let $x \in M_\calR(v,h)$. The proof consists of three steps. 

\paragraph{Step 1.} We first show that $x_i \ge \lfloor \frac{v_i h}{v_+} \rfloor$ for all parties $P_i$. Assume for contradiction that this is not the case and let $P_i$ be a party with $x_i < \lfloor \frac{v_i h}{v_+} \rfloor$. We infer that $x_i \leq \frac{v_i h}{v_+} - 1$ and, by the pigeonhole principle, that there exists a party $P_j$ such that $x_j > \frac{v_j h}{v_+}$. Thus, $u(i,x) = x_i \frac{v_+}{h}$ and $u(j,x) = v_j$.
Therefore, 
\begin{align*}
u(x_\move{j}{i}) - u(x) &= u(i,x_\move{j}{i}) + u(j,x_\move{j}{i}) - u(i,x) - u(j,x) \\
             &=  (x_i + 1)\frac{v_+}{h} + \min \left(v_j, (x_j - 1)\frac{v_+}{h} \right) -  v_j - x_i \frac{v_+}{h} \\
             &= \frac{v_+}{h} + \min \left(v_j,(x_j - 1)\frac{v_+}{h} \right) -  v_j
			 \ge \frac{v_+}{h} > 0 \text,
\end{align*}
contradicting our assumption that $x \in M_\calR(v,h)$.

\paragraph{Step 2.} We next show that $x_i \le \lfloor \frac{v_i h}{v_+} \rfloor$ for all parties $P_i$. Assume for contradiction that this is not the case and let $P_i$ be a party with $x_i > \lceil \frac{v_i h}{v_+} \rceil$. Similarly as before, we infer that $x_i \geq \frac{v_i h}{v_+} + 1$ and that there exists a party $P_j$ with $x_j < \frac{v_j h}{v_+}$. Thus, $u(i,x) = v_i$ and $u(j,x) = x_j \frac{v_+}{h}$. 
Therefore,
\begin{align*}
u(x_\move{i}{j}) - u(x) &= u(i,x_\move{i}{j}) + u(j,x_\move{i}{j}) - u(i,x) - u(j,x) \\ 
			 &= v_i + \min \left(v_j, (x_j + 1) \frac{v_+}{h} \right) - v_i - x_j \frac{v_+}{h} \\
			 &= \min \left(v_j, (x_j + 1) \frac{v_+}{h} \right) - x_j \frac{v_+}{h} \ge \frac{v_+}{h} > 0 \textrm{.}
\end{align*}
As in step 1, we have reached a contradiction.

\paragraph{Step 3.}
In the first two steps, we have shown that $\lfloor \frac{v_\ell h}{v_+} \rfloor \leq x_\ell \leq \lceil \frac{v_\ell h}{v_+} \rceil$ for each party $P_\ell$. Now we show that $M_\calR$ coincides with the \hamilton. For the sake of contradiction, assume that this is not the case. Then, there exist two parties $P_i$ and $P_j$ such that $x_i = \lfloor\frac{v_i h}{v_+}\rfloor +1$, $x_j = \lfloor\frac{v_j h}{v_+}\rfloor$, but the remainder of party $P_i$ is strictly smaller than the remainder of party $P_j$, i.e., 
$\frac{v_i h}{v_+} - \lfloor\frac{v_i h}{v_+}\rfloor < \frac{v_j h}{v_+} - \lfloor\frac{v_j h}{v_+}\rfloor$. It follows that
	\begin{equation}\label{eq:remainder}
	    \frac{v_i h}{v_+} - (x_i - 1) < \frac{v_j h}{v_+} - x_j \text.	
	\end{equation} 
Therefore, 
\begin{align*}
	u(x_\move{i}{j}) - u(x) &= u(i,x_\move{i}{j}) + u(j,x_\move{i}{j}) - u(i,x) - u(j,x) \\ 
					&= (x_i - 1) \frac{v_+}{h} + v_j - v_i - x_j \frac{v_+}{h} \\
					&= \frac{v_+}{h} \left(\frac{v_j h}{v_+} - x_j - \frac{v_i h}{v_+}  + x_i - 1 \right) > 0 \text,
\end{align*}
where the inequality is due to (\ref{eq:remainder}). We again obtain a contradiction, completing the proof.
\end{proof}

An immediate consequence of \thmref{thm:monroeApportionment} is that Monroe's rule respects quota.

\subsection{Phragm\'{e}n's Variance-Minimizing Rule as an Extension of the \sainte}

We now turn to the \sainte. Since the \sainte is the divisor method based on $(1,3,5,\ldots)$, 
\thmref{thm:owa-divisor} implies that it is induced by the Thiele rule with weight sequence $(1, \nicefrac{1}{3}, \nicefrac{1}{5}, \ldots)$.\footnote{The fact that the Thiele rule based on weight sequence $(1, \nicefrac{1}{3}, \nicefrac{1}{5}, \ldots)$ induces the \sainte in the party-list setting is also stated (without proof) by \citet{Pere16a} and \citet{BKP17a}.} However, this is not the only multiwinner rule inducing the \sainte. Recall that \varP is the variant of Phragm\'{e}n's load-balancing methods that minimizes the variance of the loads. 

\begin{theorem}
	The apportionment method induced by \varP coincides with the \sainte. 
\end{theorem}

\begin{proof}
	\citet{Sain10a} and \citet{Owen21a} have shown that the \sainte selects exactly those seat distributions $x$ minimizing 
	\[ e(x) = \sum_{i \in [p]} v_i \left( \frac{x_i}{v_i} - \frac{h}{v_+} \right)^2  = \sum_{i \in [p]} \frac{x_i^2}{v_i} - \frac{h^2}{v_+} \text. \]
	%
	%
	\citep[see also][pp.~103--104]{BaYo82a}.
	By ignoring constants, one can see that minimizing $e(x)$ is equivalent to minimizing $\sum_{i \in [p]} \frac{x_i^2}{v_i}$. 
	
	The multiwinner rule \varP chooses a committee corresponding to a load distribution $L$ minimizing $\sum_{j \in N} (\sum_{c \in C} \ell_{j,c})^2$. For a given seat distribution $x$, we  can without loss of generality assume that the load of $x_i$ is distributed equally among all $v_i$ voters in $N_i$. Thus, the total load assigned to a voter $j \in N_i$ equals $\sum_{c \in C} \ell_{j,c} = \frac{x_i}{v_i}$. Therefore,
	\[ \sum_{j \in N} \left(\sum_{c \in C} \ell_{j,c} \right)^2 = \sum_{i \in [p]} v_i \left(\frac{x_i}{v_i} \right)^2 = \sum_{i \in [p]} \frac{x_i^2}{v_i} \text. \]
	It follows that both methods solve identical minimization problems, and thus coincide.   
\end{proof}

\subsection{Other Multiwinner Rules in the Context of Apportionment}

In this section we examine two further Thiele rules, the Chamberlin--Courant rule and the {top-$k$} rule, in the context of apportionment. Since $w_{\text{CC}} = (1, 0, 0, \ldots)$ and $w_{\text{top-$k$}} = (1, \ldots, 1, 0, 0, \ldots)$, \propref{prop:sequential} applies to both rules. 

The Chamberlin--Courant apportionment method produces lots of ties, because it does not strive to give more than one seat even to parties with large support. In the special case where the number~$p$ of parties exceeds the number~$h$ of seats, the Chamberlin--Courant rules assigns one seat to each of the $h$ largest parties.\footnote{In this special case, the apportionment method induced by the Chamberlin--Courant rule coincides with impervious divisor methods (see Footnote~\ref{fn:impervious}).} 

\begin{proposition}\label{prop:ccApportionment}
Consider an instance $(v,h)$ of the apportionment problem. If $p>h$, the apportionment method induced by the Chamberlin--Courant rule assigns one seat to each of the $h$ parties with the greatest number of votes. 
\end{proposition}

%

Proposition~\ref{prop:ccApportionment} gives an interesting insight into the nature of the Chamberlin--Courant apportionment method: it selects an assembly having one representative from each of the $h$ largest groups of a given society. 

The top-$k$ rule is on the other end of the spectrum of apportionment methods.


\begin{proposition}\label{prop:topKapportionment}
The apportionment method induced by the top-$k$ rule assigns all seats to the party (or parties) receiving the highest number of votes.
\end{proposition}


Thus, the top-$k$ rule induces the apportionment method $M^*$ in Proposition 4.5 of \citet{BaYo82a}. 
This apportionment method is often used to select (minority) governments. A minority cabinet is usually formed by the party that receives the greatest number of votes, even if it does not have a majority of the seats. 

We conclude this section by mentioning two further rules that are often studied in the context of multiwinner elections. 
\emph{Satisfaction approval voting (SAV)} \citep{BrKi14a} induces the same apportionment methods as the top-$k$ rule.%
\footnote{\citet{BrKi14a} also propose an apportionment method based on maximizing the satisfaction of voters. This method, called \emph{SAV apportionment}, is different from the apportionment method $M_{\mathit{SAV}}$ that is induced by SAV via our generic three-step procedure defined at the beginning of \secref{sec:apportionmentConstruction}. } 
And \emph{minimax approval voting} \citep{BKS07a} induces the apportionment method that assigns the seats in a way that maximizes the number of seats of the party with the fewest seats. In particular, if there are more parties than seats, then all committees are winning. 

\section{Degressive Proportionality and Election Thresholds}
\label{sec:degressive}

In this section, we show that Thiele rules can also be used to induce apportionment methods with appealing properties other than lower quota. We do this by exploring \emph{degressive proportionality}, the concept suggesting that smaller populations should be allocated more representatives in comparison to the quota-based allocation methods, and \emph{election thresholds}, the concept saying that a party should only be represented in parliament if it receives at least a certain fraction of votes. To the best of our knowledge, the Thiele rules considered in this section have not been studied before. 

\subsection{Degressive Proportionality}

Despite its mathematical elegance and apparent simplicity, the proportionality principle is not always desirable. For instance, in an assembly where decisions are taken under simple majority rule, a cohesive group of $51\%$ has obviously more than a ``fair'' share.
Principles of justice indicate that, for decisions bodies governed by majority rule, fair apportionment should follow a norm of degressive proportionality \citep{laslier2012WhyNotProportional,MT12,KLMT13a}. 
In fact, degressive proportional apportionments can often be observed in parliaments that gather districts, regions, or states of very different sizes, such as the European Parliament.%
\footnote{Even though the composition of the European Parliament is not the result of the application of a well-defined rule, it is interesting to note that the history of successive negotiations produced such a result \citep[see][]{rose2013RepresentingEuropeans}.} 

The \emph{Penrose apportionment method} (aka the \emph{square-root method} and devised by \citet{Penr46a}; see also the work of \citet{SlZy06a}) allocates seats in such a way that the number of seats allocated to a party is proportional to the square root of the votes for that party.
It has been proposed
for a United Nations Parliamentary Assembly~\citep{Bumm10a} and for allocating voting weights\footnote{The question of allocating voting weights to representatives is formally equivalent to the apportionment problem.} in the Council of the European Union~\citep{euVoting}. 

Let $\zeta \colon \reals_{>0} \to \reals_{>0}$ be an increasing function. We say that an apportionment method $M$ satisfies $\zeta$-\emph{proportionality} if, for every instance $(v,h)$ and for every $x \in M(v,h)$, it holds that $x_i \ge \left\lfloor h \frac{\zeta(v_i)}{\sum_\ell \zeta(v_\ell)}\right\rfloor$ for all $i \in [p]$. Intuitively, degressively proportional apportionment methods satisfy \mbox{$\zeta$-proportionality} for some concave function $\zeta$. In particular, we say that an apportionment method satisfies the \emph{Penrose condition} if it satisfies $\zeta_P$-proportionality, where $\zeta_{P}(x) = \sqrt{x}$.

Recall that a function \mbox{$f \colon \reals \to \reals$} is \emph{multiplicative} if for each $x, y \in \reals$ it holds that $f(xy) = f(x)f(y)$.\footnote{If a multiplicative function $f$ is bijective, then its inverse $f^{-1}$ is also multiplicative.} We have the following theorem. 


\begin{theorem}\label{thm:general_proportionality}
Let $\zeta \colon \reals_{>0} \to \reals_{>0}$ be an increasing multiplicative bijection, and let $\alpha = \zeta^{-1}$ be the inverse of $\zeta$. Then, the apportionment method induced by the Thiele rule with weight sequence $w = (\nicefrac{1}{\alpha(1)}, \nicefrac{1}{\alpha(2)}, \nicefrac{1}{\alpha(3)}, \ldots)$ satisfies $\zeta$-proportionality.
\end{theorem}

\begin{proof}
	Let $\calR$ be $w$-Thiele with $w_j = \frac{1}{\alpha(j)}$ for all $j \in \naturals$. We show that $M_\calR$ satisfies \mbox{$\zeta$-proportionality}.  Assume for contradiction that for some instance $(v,h)$ of the apportionment problem, there is a seat distribution $x \in M_\calR(v,h)$ with  $x_i < \left\lfloor h \frac{\zeta(v_i)}{\sum_\ell \zeta(v_\ell)}\right\rfloor$ for some $i \in [p]$. 
	We infer that $x_i \leq h \frac{\zeta(v_i)}{\sum_\ell \zeta(v_\ell)} - 1$, and, using the pigeonhole principle, that there exists a party $P_j$ such that $x_j > h \frac{\zeta(v_j)}{\sum_\ell \zeta(v_\ell)}$. We conclude that $\frac{x_i + 1}{\zeta(v_i)} < \frac{x_j}{\zeta(v_j)}$. Since $\alpha=\zeta^{-1}$ is multiplicative, we get 
\begin{align*}
\frac{\alpha(x_i + 1)}{v_i} = \frac{\alpha(x_i + 1)}{\alpha(\zeta(v_i))} = \alpha\left(\frac{x_i + 1}{\zeta(v_i)}\right) < \alpha\left(\frac{x_j}{\zeta(v_j)}\right) = \frac{\alpha(x_j)}{\alpha(\zeta(v_j))} = \frac{\alpha(x_j)}{v_j}  \text{,}
\end{align*}
and thus $\frac{v_i}{\alpha(x_i + 1)} > \frac{v_j}{\alpha(x_j)}$.
Therefore,
\begin{align*}
u(x_\move{j}{i}) - u(x) = \hat u(i,x_i + 1) - \hat u(j,x_j) 
			 = \frac{v_i}{\alpha(x_i+1)} - \frac{v_j}{\alpha(x_j)}  > 0 \text, \end{align*}
contradicting the assumption that $x \in M_\calR(v,h)$. This completes the proof. 
\end{proof}

An interesting corollary of Theorem~\ref{thm:general_proportionality} is that the Penrose condition can be satisfied by using a Thiele rule. The weight sequence $w = (w_1, w_2, \ldots)$ achieving this is given by $w_j = \frac{1}{j^2}$. 

\begin{corollary}
The apportionment method induced by the Thiele rule with weight sequence $w = (1, \nicefrac{1}{4}, \nicefrac{1}{9}, \ldots)$ satisfies the Penrose condition.
\end{corollary}


Degressive proportionality is also an important feature of the so-called \textit{Cambridge Compromise}, which proposes an apportionment method for the European Parliament based on an affine formula: each member state should be endowed with a fixed number (5) of delegates plus a variable number, proportional to the population of the state \citep{cambridgeCompromise}. 

We show that such an apportionment method can be implemented via a Thiele rule. 

\begin{proposition}
Consider an instance $(v,h)$ of the apportionment setting with $h \ge 5p$ and let $Z$ be a constant with $Z > 5v_+$. Let $\calR$ denote $w$-Thiele with weight sequence $w = (0, 0, 0, 0, Z, 1, \nicefrac{1}{2}, \nicefrac{1}{3}, \ldots)$. Then, for every $x \in M_\calR(v,h)$ and for every $i \in [p]$,  
\[ x_i \ge 5 \quad \text{and} \quad x_i - 5 \geq \left\lfloor \frac{v_i(h - 5p)}{v_+} \right\rfloor \text. \]
\end{proposition}

\begin{proof}
We first prove that each party gets at least five seats. Indeed, for the sake of contradiction let us assume that party $P_i$ gets less than five seats. By transferring some seats from other parties (having more than five representatives) to $P_i$ we increase the satisfaction of supporters of $P_i$ by at least $v_iZ > 5v_+$. At the same time, transferring a seat from some party $P_j$ (having more than five seats) to $P_i$ reduces the satisfaction of the supporters of $P_j$ by at most $v_+$. Thus, such a transfer of at most five seats improves the total satisfaction of the voters, a contradiction.

We now show that $x_i - 5 \geq \lfloor \frac{v_i(h - 5p)}{v_+} \rfloor$ for all $x \in M_\calR(v,h)$ and for all $i \in [p]$. For the sake of contradiction, assume that there is an instance $(v,h)$ of the apportionment problem and a seat distribution $x \in M_\calR(v,h)$ with $x_i - 5 < \lfloor \frac{v_i(h - 5p)}{v_+} \rfloor$ for some $i \in [p]$. Thus, $x_i - 5 \leq \frac{v_i(h - 5p)}{v_+} - 1$, and, by the pigeonhole principle, there exists a party $P_j$ such that $x_j - 5 > \frac{v_j(h - 5p)}{v_+}$. Consequently, $\frac{x_i - 4}{v_i} < \frac{x_j - 5}{v_j}$, and thus $\frac{v_i}{x_i - 4} > \frac{v_j}{x_j - 5}$.
Therefore,
\begin{align*}
u(x_\move{j}{i}) - u(x) = \hat u(i,x_i + 1) - \hat u(j,x_j) 
	 		 = \frac{v_i}{x_i-5 + 1} - \frac{v_j}{x_j-5} > 0 \text,
\end{align*}
contradicting the assumption that $x \in M_\calR(v,h)$. This completes the proof. 
\end{proof}

\subsection{Election Thresholds}
\label{sec:thresholds}

Thiele rules also provide an interesting way for implementing election thresholds. Thresholds in the form of \emph{percentage hurdles} are often encountered in parliamentary elections \citep[see][Section 7.6]{Puke14a}. 
For a weight sequence $w = (w_1, w_2, \ldots)$, let $w[t]$ be the weight sequence obtained from $w$ by replacing the first $t$ elements of $w$ with zeros. The rule $w[t]$-Thiele induces an apportionment method that allocates seats to a party only if the number of votes the party receives exceeds a certain fraction of the total number of votes \emph{received by parties with allocated seats}.\footnote{
The thresholds implemented by Thiele-based apportionment methods differ from real-world election thresholds in one important aspect: Rather than requiring a certain fraction (say, 5\%) of \emph{all} votes, the thresholds we consider here require a certain fraction of those votes that are received by parties \emph{with allocated seats}.
An advantage of the latter kind of threshold requirement is that it can always be satisfied. By contrast, requiring a certain fraction of all votes can lead to situations in which \emph{no} party reaches the required number of votes and, therefore, no seat allocation satisfies the threshold requirement.  
}
Proposition~\ref{prop:pavThreshold} formalizes this behavior for the weight sequence $w_{\text{PAV}}$.

\begin{proposition}\label{prop:pavThreshold}
	Consider an instance $(v,h)$ of the apportionment problem and 
fix an integer $t$ with $1 \le t < h$.  
Let $\calR$ denote the Thiele rule with weight sequence $w_\text{PAV}[t]$. Then, for all seat distributions $x \in M_\calR(v,h)$, 
\[x_i >0 \quad \text{only if}  \quad \frac{v_i}{\sum_{\ell \in [p]: x_\ell>0} v_\ell} > \frac{t}{h} \text. \] 
\end{proposition}




\begin{proof}
	Consider a seat distribution $x \in M_\calR(v,h)$.
	First, observe that there is at least one party that is allocated strictly more than $t$ seats. (Otherwise, the total voter satisfaction $u(x)$ would be zero and could be improved upon by giving all seats to an arbitrary party.) Without loss of generality let us assume that party $P_1$ is such a party, \ie $x_1>t$.
	
	Second, observe that \emph{every} party that is allocated \emph{some} seat under $x$, is allocated strictly more than $t$ seats. In other word, $x_\ell > 0$ implies $x_\ell >t$. (Otherwise, the total voter satisfaction provided by $x$ could be improved upon by $x_\move{\ell}{1}$.)


Now consider a party $P_i$ with $x_i>0$. The argument above implies that $x_i >t$.
Let $v_+^x$ denote the total number of votes for parties with allocated seats under~$x$, \ie $v_+^x = \sum_{\ell \in [p]: x_\ell>0} v_\ell$. 
If $\frac{v_i}{x_i} \ge \frac{v_+^x}{h}$, then $\frac{v_i}{v_+^x} \geq \frac{x_i}{h} > \frac{t}{h}$ and we are done. So let us assume that $\frac{v_i}{x_i} < \frac{v_+^x}{h}$. 
From the pigeonhole principle we infer that there exists a party $P_j$ with $x_j>0$ such that $\frac{v_j}{x_j} \ge \frac{v_+^x}{h}$. Since $x_j >0$, the observation above implies that $x_j > t$. Therefore,
\begin{align*}
0 \ge u(x_\move{i}{j}) - u(x) = \hat u(j,x_j+1) - \hat u(i,x_i) = \frac{v_j}{x_j + 1} - \frac{v_i}{x_i}
\end{align*}
and thus 
\begin{align*}
v_i \geq x_i \frac{v_j}{x_j + 1} = x_i \frac{v_j}{x_j} \cdot \frac{x_j}{x_j + 1} > x_i \frac{v_j}{x_j} \cdot \frac{t}{t + 1} \geq t \frac{v_j}{x_j} \geq t \frac{v_+^x}{h} \text.
\end{align*}
This completes the proof.
\end{proof}

An interesting open question is whether it is possible to naturally modify the definition of the Monroe system to implement election thresholds.

\section{Conclusion}
\label{sec:conclusion}

In legislative procedures, proportional representation is typically achieved by employing party-list apportionment methods such as the \dhondt or the \hamilton. These methods impose proportionality by assuring that each party in a representative body is allocated a number of representatives that is proportional to the number of received votes. 
Nevertheless, party-list systems have many drawbacks---for instance, they provide very weak links between elected legislators and their constituents.

In this paper we have proposed a simple and natural formal framework that allows us to view approval-based multiwinner election rules as apportionment methods. Some multiwinner rules (such as PAV, Monroe's rule, \phrag's rules, and the rules considered in \secref{sec:degressive}) induce apportionment methods that are used---or have been proposed---for parliamentary apportionment as variants of the notions of proportional, or degressive proportional, representation. 
The apportionment methods that are induced by other multiwinner rules (such as the Chamberlin--Courant rule, the top-k rule, Satisfaction Approval Voting, and Minimax Approval Voting) are, however, less appealing. These results give interesting insights into the nature of the analyzed multiwinner rules, and they show how traditional apportionment methods can be extended to settings where voters can vote for individual candidates.
Further, our methodology allowed us to discover a new interesting multiwinner election rule which, if applied to the apportionment setting, satisfies Penrose's (square root proportionality) condition.

\paragraph{Acknowledgments}

We thank Bernard Grofman, Svante Janson, and Friedrich Pukelsheim for helpful discussions. We also thank the anonymous reviewers for their comments, which helped improve the paper. This material is based on work supported by ERC-StG 639945 (ACCORD), by a Feodor Lynen return fellowship of the Alexander von Humboldt Foundation, by a Humboldt Research Fellowship for Postdoctoral Researchers, and by the ANR project ANR13-BSH1-0010 DynaMITE. A preliminary version of this paper appeared in the proceedings of the \emph{31st Conference on Artificial Intelligence (AAAI-2017)}.

\bibliographystyle{plainnat}
\bibliography{main}

\appendix

\section{Proportional Justified Representation Implies Lower Quota}
\label{sec:jr}

In this section, we prove that the every multiwinner rule satisfying proportional justified representation induces an apportionment method satisfying quota. The following definition is due to \citet{SFF+17a}.

\begin{definition}
	Consider an instance $(A,k)$ of the multiwinner election setting and a size-$k$ committee $S$. The committee $S$ \emph{provides proportional justified representation (PJR) for $(A,k)$} if there does not exist a subset  $N^* \subseteq N$ of voters and an  integer $\ell>0$ with $|N^*| \ge \ell \frac{n}{k}$ such that
	$| \bigcap_{i \in N^*} A_i| \ge \ell$, but $|S \cap (\bigcup_{i \in N^*} A_i)|<\ell$. 
	A multiwinner rule $\calR$ \emph{satisfies proportional justified representation (PJR)} if for every instance $(A,k)$, every committee in $\calR(A,k)$ provides PJR for $(A,k)$.
\end{definition}

\begin{proposition}
	If a multiwinner rule satisfies proportional justified representation, then it satisfies lower quota. 
\end{proposition}

\begin{proof}
	Let $\calR$ be a multiwinner election rule satisfying PJR, and consider an instance $(v,h)$ of the apportionment problem. Consider an arbitrary party $P_i$ with $i \in [p]$. Translated to the multiwinner election setting, there is a group $N_i$ of voters of size $|N_i| = v_i$ and every voter in $N_i$ approves all $h$ candidates in $C_i$. Define $\ell = \lfloor \frac{v_i h}{v_+} \rfloor$. We have $\ell \le \frac{v_i h}{v_+}$ and thus $v_i \ge \ell \frac{v_+}{h}$.
	
	Let $N^* = N_i$. Since $|N^*| = v_i \ge \ell \frac{v_+}{h}$ and $| \bigcap_{i \in N^*} A_i| = |C_i| = h \ge \ell$, proportional justified representation implies that every committee $S$ that is chosen by $\calR$ satisfies $|S \cap (\bigcup_{i \in N^*} A_i)| \ge \ell$. And since $S \cap (\bigcup_{i \in N^*} A_i) = S \cap C_i$, we know that $x_i = |S \cap C_i| \ge \ell$. Thus, party $P_i$ is allocated at least $\ell = \lfloor \frac{v_i h}{v_+} \rfloor$ seats, as desired.
\end{proof}

Since Proportional Approval Voting satisfies PJR---and in fact even the stronger property \emph{extended justified representation} \citep{ABC+17a}---it also satisfies lower quota.


\end{document}